\documentclass{article}
\usepackage[utf8]{inputenc}

\title{Randomized Privacy Budget Differential Privacy}
\author{Meisam Mohammady }
\date{December 2018}
\usepackage{booktabs}
\usepackage{amsmath}
\usepackage{graphicx}
\usepackage{algorithm2e}
\usepackage{mathtools}
\usepackage[thinc]{esdiff}
\usepackage{courier}
\usepackage{amsthm}
\usepackage{adjustbox,multirow}
\newcommand{\subhead}[1]{\vspace {0.04in}\noindent{\textbf{#1.}}}

\newcommand{\DP}{differential privacy }

\usepackage{courier}
\usepackage{amsthm}
\newtheorem{thm}{Theorem}[section]
\newtheorem{lem}[thm]{Lemma}

\newtheorem{coro}[thm]{Corollary}

\newtheorem{defn}{Definition}[section]
\usepackage{bbm, dsfont}

\usepackage{cite}
\usepackage{amsmath}
\usepackage{amssymb}
\newcommand{\Adj}{\text{Adj}}

\newcommand{\RN}[1]{%
  \textup{\uppercase\expandafter{\romannumeral#1}}%
}
\newcommand\norm[1]{\left\lVert#1\right\rVert}

\newcommand{\D}{\mathsf D}
\newcommand{\R}{\mathsf R}

\newcommand{\Prob}{\mathbb P}
\usepackage{amssymb}
\begin{document}
\maketitle
\section{Abstract}
While pursuing better utility by discovering knowledge from the data, individual's privacy may be compromised during an analysis. To that end, differential privacy has been widely recognized as the state-of-the-art privacy notion. By requiring the presence of any individual's data in the input to only marginally affect the distribution over the output, differential privacy provides strong protection against adversaries in possession of arbitrary background. However, the privacy constraints (e.g., the
degree of randomization) imposed by differential privacy may render
the released data less useful for analysis, the fundamental trade-off
between privacy and utility (i.e., analysis accuracy) has
attracted significant attention in various settings. In this report we present DP mechanisms with randomized parameters, i.e., randomized privacy budget, and formally
analyze its privacy and utility and demonstrate that randomizing privacy budget in DP mechanisms will boost the accuracy in a humongous scale.
\section{Backgrounds}
\begin{defn}
Let $\epsilon, \delta \geq 0$. A mechanism $M: \D \times \Omega \to \R$ is 
$(\epsilon, \delta)$-differentially private for $\Adj$ if for all $d,d' \in \D$ such that $\Adj(d,d')$, we have
\begin{align}	\label{eq: standard def approximate DP original}
\Prob(M(d) \in S) \leq e^{\epsilon} \Prob(M(d') \in S) + \delta, \;\; \forall S \in \mathcal M. 
\end{align}
If $\delta=0$, the mechanism is said to be $\epsilon$-differentially private. 
\end{defn}
 
\begin{defn}(Usefulness Definition).
\label{defn:useful}
A database
mechanism $M_q$ is ($\zeta,\gamma(\zeta,u)$)-useful if with probability
$1 - \gamma(\zeta,u)$, for every database $d \subseteq \D$, $|M_q(d)- q(d)| \leq \zeta$.
\end{defn}

\section{Randomized Parameter DP}

Let $M_q(d,u) = q(d) \bigoplus \omega(u)$ be a randomized ($\epsilon(u),\delta(u)$)-differentially private mechanism where $\omega(u)$ is a random oracle with specified set of parameters $u$ and $\bigoplus$ stands for the corresponding operator. Also suppose $M_q(d,u)$ is ($\zeta,\gamma(\zeta,u)$)-useful. Define by  $\mathcal M_q(d) = q(d) \bigoplus \omega(u)$, with $u \sim \mathcal F$, the distribution of all possible randomized mechanism $M_q(d,u)$ where $\mathcal  F$ is a probability density function for all parameters in $u$.
The optimal utility achieved due to the application of optimal pdf $\mathcal F$ is shown in the following.
\begin{equation}
\label{gen:ut}
       U(\zeta)= Max \ \{E\left[\Prob(|M_q(d,u)-q(d)|)<\zeta\right]\}
\end{equation}
Accordingly, we say that $\mathcal M_q(d)$ improves the privacy-utility trade-off if we have
\begin{itemize}
    \item \textbf{Case \RN{1} ($\delta=0$)}
    \begin{equation}
    \label{gen:conddelt0}
        E(e^u)= e^{\epsilon(u_0)} \Rightarrow U(\zeta)>1-\gamma(\zeta,u_0)
    \end{equation}
    over
    \item  \textbf{Case \RN{2} ($\delta>0$)}
    \begin{eqnarray}
      E(\delta(u))= \delta(u_0) & \Rightarrow U(\zeta)>1-\gamma(\zeta,u_0) \\ & E(e^u)< e^{\epsilon(u_0)}
    \end{eqnarray}
\end{itemize}  
where, $E(.)$ denotes the expected value over distribution $\mathcal F$.
We now derive the corresponding conditions for two popular differentially private mechanisms. In particular, a \emph{Laplace Mechanism} modifies an answer to a numerical query by adding independent
and identically distributed (i.i.d.) zero-mean noise distributed~\cite{textbook}, \cite{lamport93}, \cite{lamport94} according to a Laplace distribution. Recall that the Laplace distribution
with mean zero and scale parameter $b$, denoted $Lap(b)$,
has density $p(x;b)=\frac{1}{2b}exp(-\frac{|x|}{b})$ and variance $2b^2$ . Moreover, for $\omega \in \mathbb R^k$ with $\omega_i$ i.i.d. and $\omega_i \sim Lap(b)$ , denoted $\omega \sim Lap(b)^k$, we have $p(\omega;b)=(\frac{1}{2b})^k exp(-\frac{ \norm \omega_1 }{b})$, $E( \norm \omega_1)=b$, and $\Prob( \norm \omega_1 \geq tb)=e^{-t}$. 
\begin{thm}	\label{thm: Lap mech}
Let $q: \D \to \mathbb R^k$ be a query , $epsilon>0$.
Then the  mechanism $M_q: \D \times \Omega \to \mathbb R^k$ 
defined by $M_q(d) = q(d) + w$, with $w \sim Lap(b)^k$, 
where $b \geq \frac{\Delta_1 q}{\epsilon}$
is $\epsilon$-differentially private.
\end{thm}
Hence, $\omega(u)$ is a Laplace distribution where $u=1/b=\epsilon$. Also, $\gamma(\zeta,u_0)=e^{-\zeta \epsilon}$. Thus, equations~\ref{gen:ut},\ref{gen:conddelt0} can be re-written as follows.  
\begin{equation}
   U(\zeta)= Max \ E(1-e^{-\zeta \epsilon})
\end{equation}
\begin{equation}
    \label{gen:conddelt0}
        E(e^\epsilon)= e^{\epsilon_0} \Rightarrow U(\zeta)>1-e^{-\zeta \epsilon}
    \end{equation}
Similarly, for a Gaussian mechanism, we have

\begin{equation}
   Min \ E(e^{\epsilon})=\int^{\infty}_{-\infty} f(\epsilon) e^{\epsilon}  d\epsilon
\end{equation}
over 
\begin{equation}
 E(Q(\frac{\zeta}{\sqrt{\frac{2\zeta+1}{2\epsilon}}}))=\int^{\infty}_{-\infty} f(\epsilon) Q(\frac{\zeta}{\sqrt{\frac{2\zeta+1}{2\epsilon}}})   d \epsilon \leq \delta
\end{equation}

\section{Privacy and Utility Analysis}
\label{sec5}
In this section, we formally
characterize the privacy and the utility of the Randomized DP mechanism. 

\subsection{Deriving PDF of Randomized DP}

we can write the CDF of the output of an Randomized DP Laplace mechanism in terms of the \emph{Moment Generating Function (MGF)} for the probability distribution $f_{\frac{1}{b}}$. Recall that MGF of a random variable is an alternative specification of its
probability distribution, and hence provides the basis of an
alternative route to analytical results compared with working directly
with probability density functions or cumulative distribution
functions. In particular,
\begin{defn}(\textit{Moment Generating Function})
The moment-generating function of a random variable $x$ is $M_{X}(t):=\mathbb E \left[e^{tX}\right], t\in \mathbb {R}$ wherever this expectation exists. The moment-generating function is the expectation of the random variable $e^{tX}$.
\end{defn}
Accordingly, in the following, we give a general formula for the probability of any measurable event originated from an Randomized DP Laplace Mechanism.
\begin{thm}	\label{thm: RPLap mech}
The search space of an Randomized DP Laplace mechanism is as large as the space of all PDFs with non-negative support and existing MGF. Moreover, generated PDFs are all log-convex.
\end{thm}
\vspace{0.05in}
Thus, for a PDF with non-negative support (scale parameter is always non-negative), the Randomized DP Laplace mechanism outputs another
PDF using the MGF (CDF is the moment and PDF is its derivative) as shown in Equation~11 in Appendix \cite{lamport93}. However, a challenge is that not all random variables have moment generating functions (MGFs). Fortunately, MGFs possess an appealing composability property between
independent probability distributions, which can be used to provide us with a search space of all linear combinations of a set of popular distributions with known MGFs (infinite number of RVs). 
\begin{thm}[MGF of Linear Combination of RVs]
\label{thm:lin}
 If $x_1, x_2$, $\cdots, x_n$ are $n$ independent RVs with respective MGFs $M_{x_i}(t)=\mathbb E (e^{t x_i})$ for $i = 1, 2,\cdots , n$, then the MGF of the linear combination $Y=\sum\limits_{i=1}^{n}a_ix_i$ is $\prod \limits_{i=1}^{n} M_{x_i} (a_it)$.
\end{thm}
 Thus, our search space is given as all possible  linear combinations of a set of independent RVs with existing MGF (Section~\ref{pdffind} demonstrates on how to choose the set of independent RVs).
 \begin{figure*}[!t]
\includegraphics[width=0.9\linewidth,viewport= 10 170 800 600,clip]{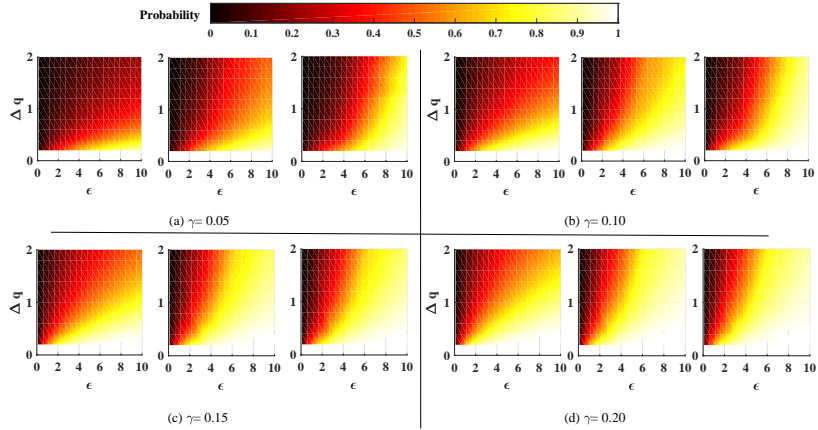}
\centering
\caption{Illustrating the performance of Randomized DP Laplace mechanism. The left, middle and right figures in each of the configurations are respectively the usefulness for the Laplace mechanism, the Randomized DP Laplace mechanism and the optimal noise, i.e., Laplace distribution in high privacy regime and Staircase shape distribution in low privacy regime}
\label{fig:both}
\end{figure*}
\begin{figure}[ht]
\includegraphics[width=0.75\linewidth,viewport= 60 -10 712 620,clip]{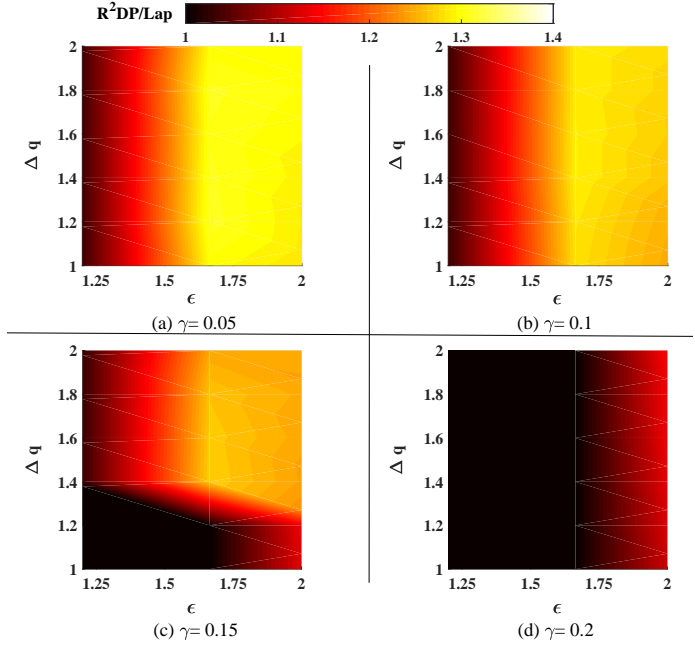}
\centering
\caption{Validating the effectiveness of Randomized DP for small $\epsilon$}
\label{fig:comp}
\end{figure}

\begin{figure}[ht]
\includegraphics[width=0.7\linewidth,viewport= 30 90 730 505,clip]{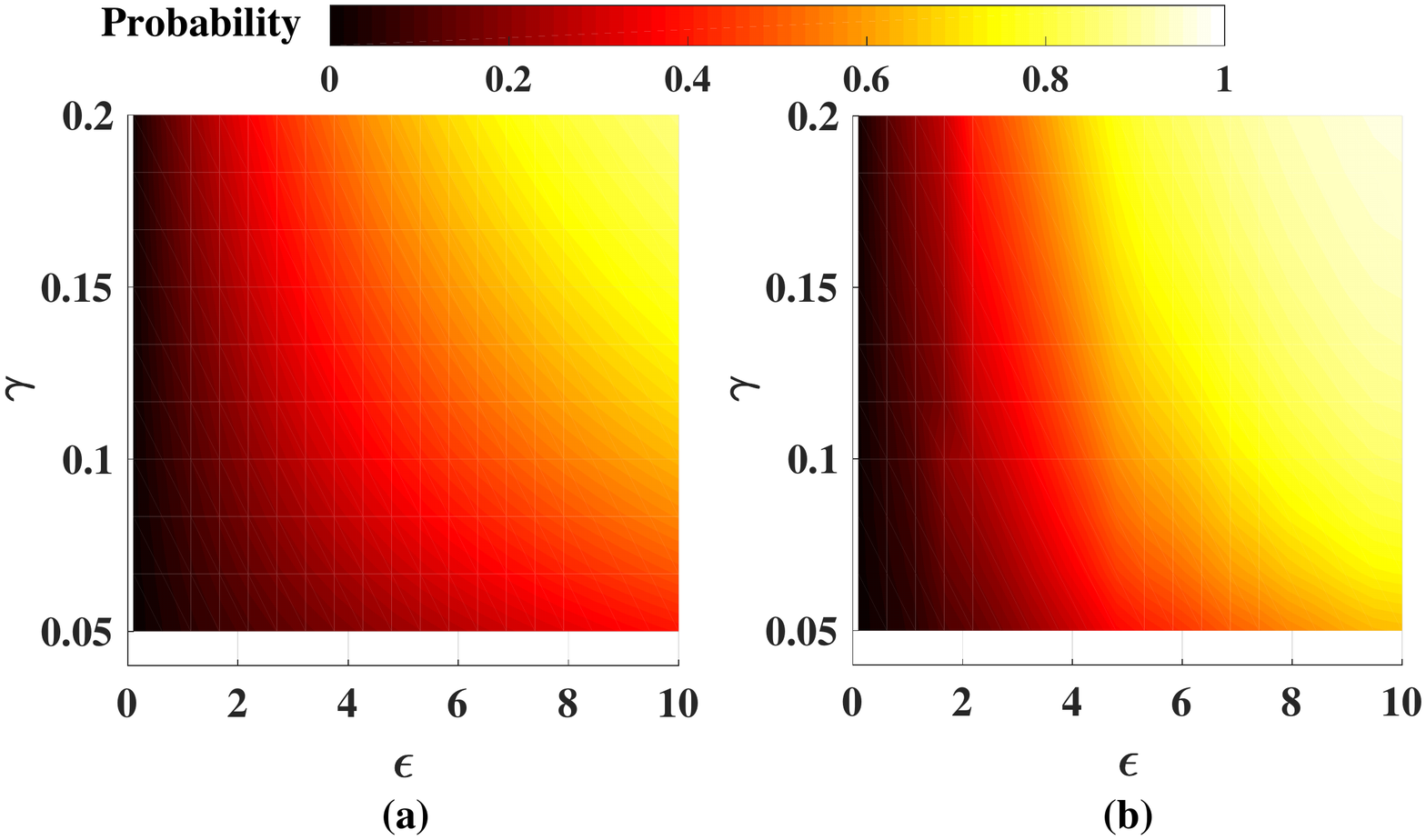}
\centering
\caption{Comparing the performance of (a) the baseline Laplace mechanisms, for count queries and (b) Randomized DP, when varying $\gamma$ and $\epsilon$}
\label{fig:bothgam}
\end{figure}
\subsubsection{Determining the Optimal PDF}
After giving the differential privacy guarantee and characterizing the utility of the Randomized DP Laplace mechanism (see Section~\ref{sec5}), we will show that the Randomized DP framework can unify two parallel concepts, i.e., privacy and utility, into one optimization problem defined over the defined search space of RVs.
\subsection{Numerical Analysis}
 We now present numerical results to fine tune the Randomized DP parameters under more general settings. In particular, Figure~\ref{fig:both} depicts the corresponding performance of Laplace mechanism, Randomized DP Laplace mechanism and Staircase mechanism. Figure~\ref{fig:both} clearly demonstrates the fact that Randomized DP can achieve both objectives mentioned earlier, i.e., approaching the optimal mechanism and improving Laplace mechanisms for larger $\epsilon$. We now analyze the improvements provided by Randomized DP under two different settings. First, we discuss the performance of Randomized DP under a stronger privacy guarantee (e.g., $\epsilon< 2$). Next, we study the improvement for
counting queries ($\Delta q = 1$) while varying the error bound $\gamma$.

\subsection{Privacy Analysis}
 We now show the Randomized DP Laplace mechanism provides differential
 privacy guarantee. Using theorem~\ref{thm: RPLap mech}, the DP bound is 
\begin{eqnarray*}
\label{DPlapexmp1}
 &\hspace{-.4cm} e^{\epsilon}= \max\limits_{\forall S \in \R} \left\{\frac{-M_{\frac{1}{b}}(-|x-q(d)|)|_{S_{\geq q(d)}}+M_{\frac{1}{b}}(-|x-q(d)|)|_{S_{< q(d)}}}{-M_{\frac{1}{b}}(-|x-q(d')|)|_{S_{\geq q(d')}}+M_{\frac{1}{b}}(-|x-q(d')|)|_{S_{< q(d')}}}\right\} 
\end{eqnarray*}
 Hence, the value of $e^{\epsilon}$ only depends on the distribution of reciprocal of the scale parameter $b$, i.e., $f_{\frac{1}{b}}$. Moreover, an MGF is positive and log-convex where the latter property is desirable in defining various natural logarithm upper-bounds, e.g., DP bound.
In the following theorem, we demonstrate the fact that our MGF-based formula for the probability $\Prob(\{q(d)+Lap(b)\}\in S)$ in Equation~\ref{lRandomized DPlap1} can be easily applied to calculate the \DP guarantee.

\begin{thm}
\label{simple DP}
The Randomized DP mechanism $\mathcal M_q(d,b)$ is 
\begin{equation}
\label{simple DPeq}
    \ln \left[ \cfrac{ \mathbb E(\frac{1}{b})} {\diff{M_{\frac{1}{b}}(t)}{t}|_{t=- \Delta q}} \right]\text{-differentially private.}
\end{equation}
\end{thm}

Finally, Theorem~\ref{thm:lin} can be directly applied to calculate the \DP guarantee of any RV from our defined search space (all linear combinations of a set of independent RVs with known MGFs).
\begin{coro}
[\DP of combined PDFs]
\label{thm:mgffin}
If $x_1, x_2, \cdots, x_n$ are $n$ independent random variables with respective MGFs $M_{x_i}(t)=\mathbb E (e^{t x_i})$ for $i = 1, 2,\cdots, n$, then the Randomized DP mechanism $\mathcal M_q(d,b)$ where $\frac{1}{b}$ is defined as the linear combination $\frac{1}{b}=\sum\limits_{i=1}^{n}a_ix_i$ is
 \vspace{-0.25cm}
 \begin{eqnarray}
 \scriptsize
    \label{mgfdp}
  ln\left[\cfrac{\sum \limits_{j=1}^{n} a_j\cdot E_{x_j}(\frac{1}{b})}{\sum \limits_{j=1}^{n} a_j\cdot M'_{x_j}(-a_j\cdot \Delta q) \cdot \prod \limits_{\substack{i=1 \\ i\neq j}}^{n} M_{x_i} (-a_i\cdot \Delta q)}\right]
    \end{eqnarray}
-differentially private.
\end{coro}
Therefore, we have established a search space of probability distributions with a universal formulation for their \DP guarantees, which is the key enabler for the universality of Randomized DP. Next, we characterize the utility of Randomized DP Laplace mechanisms.  
\subsection{Characterizing the Utility}
We now characterize the utility of the Randomized DP Laplace mechanism. To make concrete discussions, we first focus our discussion on the usefulness metric (see Section~\ref{sec:back}), then discuss how a similar logic applies to other metrics.
Denote by $U(\epsilon, \Delta q, \gamma)$ the usefulness of an Randomized DP Laplace mechanism for all $ \epsilon>0$, sensitivity $\Delta q$ and error bound $\gamma$. The optimal usefulness is then given as the answer of the following optimization problem over the search space of PDFs. 
\vspace{-0.25cm}
\begin{eqnarray}
\label{multi-obj1}
&\hspace{-0.5cm} \max\limits_{f_{\frac{1}{b}}\in F} \big\{U(\epsilon, \Delta q, \gamma) \big\}=\max\limits_{f_{\frac{1}{b}}\in F} \bigg \{\frac{1}{2} \cdot \Big[-M_{\frac{1}{b}}(-|x-q(d)|)|_{q(d)}^{q(d)+\gamma}\nonumber\\
&\hspace{3.5cm}+M_{\frac{1}{b}}(-|x-q(d)|)|_{q(d)-\gamma}^{q(d)}\Big] \bigg \}, \nonumber \\
& \text{subject to     } \ \ \ \epsilon=\ln \left[ \cfrac{ \mathbb E(\frac{1}{b})} {\diff{M_{\frac{1}{b}}(t)}{t}|_{t=- \Delta q}} \right]  \nonumber
\end{eqnarray}
Note that $\epsilon$ and $\Delta q$ do not directly impact the
usefulness but they do so indirectly through the \DP
constraint. Furthermore, as shown in Theorem~\ref{simple DP}, the \DP guarantee $\epsilon$ over the established search space $\mathcal{F}$ is a unique function of the parameters of the second fold distribution. 
\begin{coro}
\label{co1}
Denote by  $u$, the set of parameters for a probability distribution $f_{\frac{1}{b}}$, and by $M_{f(u)}$ its MGF. Then, the optimal usefulness of an Randomized DP mechanism utilizing $f_{\frac{1}{b}}$, at each triplet $(\epsilon, \Delta q, \gamma)$ is
\begin{eqnarray}
\small
    \label{gen:conddelt}
       &\hspace{-0.5cm}  U_f(\epsilon, \Delta q, \gamma)=\max\limits_{u\in \mathbb{R}^{|u|}} \bigg \{\frac{1}{2} \cdot \Big[-M_{f(u)}(-|x-q(d)|)|_{q(d)}^{q(d)+\gamma}\nonumber\\
&\hspace{3.5cm}+M_{f(u)}(-|x-q(d)|)|_{q(d)-\gamma}^{q(d)}\Big] \bigg \}, \nonumber\\
& \text{subject to     } \ \ \ \epsilon=\ln \left[ \cfrac{ \mathbb E(\frac{1}{b})} {\diff{M_{\frac{1}{b}}(t)}{t}|_{t=- \Delta q}} \right]  \nonumber
    \end{eqnarray}
\end{coro}
However, MGFs are positive and log-convex, with $M(0)=1$ and hence,  $U_f(\epsilon, \Delta q, \gamma)=1-\min\limits_{u\in \mathbb{R}^{|u|}} M_{f(u)}(-\gamma)$. Therefore, for usefulness metric, the best distribution for $\epsilon$ is the one with minimum MGF evaluated at $\gamma$. In particular, for a set of privacy and utility parameters, one can find the optimal point using the \textit{Lagrange multiplier} method. i.e.,
\begin{eqnarray}
    \label{lagrange1}
 \mathcal{L} (u,\lambda)= M_{f(u)}(-\gamma)+ \lambda \cdot (\ln \left[ \cfrac{ \mathbb E(\frac{1}{b})} {\diff{M_{\frac{1}{b}}(t)}{t}|_{t=- \Delta q}} \right]-\epsilon) 
    \end{eqnarray}
   Next, under the DP guarantee of several probability distributions, we will apply Equation~\ref{lagrange1} to find the optimal trade-off.  
    
\subhead{Utility under Other Metrics}
We derive the utility of the Randomized DP Laplace mechanism under some well-known utility metrics. Due to space limitation, we present only the final results in Table~\ref{tablemetrics}.

\begin{table}[ht]
\caption{Utility of the Randomized DP (Laplace) under different metrics}
\centering
\begin{adjustbox}{width=0.5\textwidth,center}

\begin{tabular}{|c|c|c|c|}
\hline
\bf{$\ell_1$} & \bf{ $\ell_2$}&  Entropy&  Usefulness
\\
 \hline 
$\int\limits_{0}^\infty M_{\frac{1}{b}}(-x) dx$ & $\sqrt{2\iint\limits_0^\infty  M_{\frac{1}{b}}(-u) du dx}$& $\displaystyle \int_{0}^\infty -M'_{\frac{1}{b}}(-x) \cdot \ln M'_{\frac{1}{b}}(-x)dx $ & $1-M_{\frac{1}{b}}(-\gamma)$\\
\hline
\end{tabular}
\end{adjustbox}
	\label{tablemetrics}
\end{table}
The results in Table~\ref{tablemetrics} can be easily applied to optimize each measure in different applications. 


\subhead{Necessary Condition on Selected Distributions}
Not all second fold probability distributions can boost the utility of the baseline Laplace mechanism. Accordingly, in the following theorem, we derive a necessary condition on the \DP guarantee of an Randomized DP Laplace mechanism to boost the utility of the baseline Laplace mechanism (refer to Appendix in~\cite{lamport93} for the proof). Using this necessary condition, we can easily filter out those probability distributions that cannot deliver any utility improvement.  
\begin{thm}
\label{thm: RPLap mechut}
The utility of an Randomized DP Laplace mechanism with $\epsilon\geq \ln \Big[ \mathbb E_{\frac{1}{b}} \big(e^{\epsilon(b)} \big ) \Big]$ is always upper bounded by the utility of the $\epsilon$-differentially private baseline Laplace mechanism. Equivalently, for an Randomized DP Laplace mechanism to boost the utility, the following relation is necessarily true.
\begin{equation}
\label{necc}
e^\epsilon=\frac{\mathbb E(\frac{1}{b})}{M'_{\frac{1}{b}}(- \Delta q)} < M_{\frac{1}{b}} (\Delta q)
\end{equation}
\end{thm}
We note that $\epsilon= \ln \Big[ \mathbb E_{\frac{1}{b}} \big(e^{\epsilon(b)} \big ) \Big]$ provides a tight upper bound as it gives the overall $e^\epsilon$ of an Randomized DP Laplace mechanism as the average of \DP leakages.


\subsection{Finding Utility-Maximizing Probability Distributions} 
\label{pdffind}
We now examine a set of well-known
probability distributions to establish the required search space by selecting those offer a significantly improved
$\epsilon$ compared with the bound given in Theorem~\ref{thm: RPLap mechut}. Promisingly, our analytic evaluations for
\emph{three} of these distributions, i.e., Gamma, Uniform and
truncated Gaussian distributions demonstrates such a payoff. Finally, we note that those chosen distributions are general enough to
cover a large family of other probability distributions.
For instance, since Exponential distribution, Erlang distribution,
and Chi-squared distribution are special cases
of Gamma distribution, we will only consider Gamma
distribution.

\subsubsection{Discrete Probability Distributions}

First, we consider two different mixture Laplace distributions that
can be applied for constructing Randomized DP Laplace mechanisms with discrete
probability distribution $f_b$.

\vspace{0.05in}

(1) \textbf{Degenerate distribution.} A degenerate distribution is a
probability distribution in a (discrete or continuous) space with
support only in a space of lower
dimension. If the degenerate
distribution is uni-variate (involving only a single random variable)
it will be a deterministic distribution and takes only a single
value. Therefore, the degenerate distribution is identical to the
baseline Laplace mechanism as it also assigns the mechanism one single
scale parameter $b_0$. Specifically, the probability mass function of
the uni-variate degenerate distribution is:
   \[ f_{\delta,k_0}(x)= \begin{cases} 
      1 & x= k_0  \\
      0 & x\neq k_0 
   \end{cases}
\]
 The MGF for the degenerate distribution $\delta_{k_0}$ is given by
 $M_k(t)=e^{t\cdot k_0}$. Using
 Equation~\ref{simple DPeq}, Theorem~\ref{degenerateDP} gives the
 same DP guarantee as the baseline Laplace mechanism.
 
 \begin{thm}
 \label{degenerateDP}
 The Randomized DP Laplace mechanism $M_q(d,\epsilon)$, $\epsilon \sim f_{\delta,\frac{1}{b_0}}(\epsilon)$, is $\frac{\Delta q}{b_0}$-differentially private. 
 \end{thm}
 
Obviously, this distribution does not improve the bound in
Theorem~\ref{thm: RPLap mechut} but shows the
soundness of our findings.

\vspace{0.05in}

(2) \textbf{Bernoulli distribution.}
The probability mass function of this distribution, over possible outcomes $k$, is
\[ f_{B}(k;p)=\begin{cases}p&{\text{if }}k=1,\\q=1-p&{\text{if }}k=0.\end{cases}\]
Note that the binary outcomes $k=0$ and $k=1$ can be mapped to any two outcomes $X_0$ and $X_1$, respectively. Therefore, we consider the following Bernoulli outcomes
\[ f_{B,X_0, X_1}(X;p)=\begin{cases}p&{\text{if }}X=X_1,\\q=1-p&{\text{if }}X=X_0.\end{cases}\]
The MGF for Bernoulli distribution $f_{B,X_0, X_1}(X;p)$ is  $M_X(t)=p\cdot e^{t\cdot X_0}+(1-p)\cdot e^{t\cdot X_1}$. We now derive the precise \DP guarantee of an Randomized DP Laplace mechanism with its scale parameter randomized according to a Bernoulli distribution. 
\begin{thm}
\label{bernoulidp}
 The Randomized DP Laplace mechanism $M_q(d,\epsilon)$, $\epsilon\sim f_{B,\frac{1}{b_0},\frac{1}{b_1}}(\epsilon;p)$, satisfies $\ln [p\cdot e^{\frac{\Delta q}{b_0}}+(1-p)\cdot e^{\frac{\Delta q}{b_1}}]$-DP. 
\end{thm}
However, this bound is exactly the mean value of $e^{\epsilon(b)}$ and therefore, this distribution does not improve the bound given in Theorem~\ref{thm: RPLap mechut}, either.

\subsubsection{Continuous Probability Distributions}

We now investigate three compound Laplace distributions.

\vspace{0.05in}
(1) \textbf{Gamma distribution.} The gamma distribution is a
two-parameter family of continuous probability distributions with a
shape parameter $k>0$ and a scale parameter $\theta$. Besides the
generality, the gamma distribution is the maximum entropy probability
distribution (both w.r.t. a uniform base measure and w.r.t. a $1/x$ base measure) for a random variable $X$ for which
$\mathbb E (X) = k \theta = \alpha/\beta$ is fixed and greater than
zero, and $\mathbb E[\ln(X)] = \psi(k) + \ln(\theta) =
\psi(\alpha)-\ln(\beta)$ is fixed ($\psi$ is the digamma
function). Therefore, it may provide a relatively higher
privacy-utility trade-off in comparison to the other
candidates. A random
variable $X$ that is gamma-distributed with shape $\alpha$ and rate
$\beta$ is denoted by $X\sim \Gamma(k,\theta)$ and the corresponding
PDF is
\begin{equation*}
    f_{\Gamma}(X;k,\theta)
    {\displaystyle={\frac {x^{k -1}e^{-\frac{x}{\theta}}}{\Gamma (k)\cdot \theta^k}}\quad {\text{ for }}X>0{\text{ and }}k ,\theta >0,} 
\end{equation*}
where $\Gamma ( \alpha )$ is the gamma function. We now investigate the \DP guarantee provided by assuming that the reciprocal of the scale parameter $b$ in Laplace mechanism is distributed according to the gamma distribution (see Appendix~\cite{lamport93} for the proof).
\begin{thm}
\label{gamadist}
 The Randomized DP Laplace mechanism $M_q(d,\epsilon)$, $\epsilon \sim f_{\Gamma}(\epsilon;k,\theta)$, satisfies $\big((k+1)\cdot \ln (1+\Delta q \cdot \theta)\big)$-DP. 
\end{thm}
We now apply the necessary condition given in Equation~\ref{necc}.
\begin{lem}
\label{necsgama}
   Randomized DP using Gamma distribution can satisfy the necessary condition in Equation~\ref{necc}.
\end{lem}
\begin{proof}
    We need to show that there exist $k$ and $\theta$ such that $(k+1)\cdot \ln (1+\Delta q \cdot \theta) < -k \cdot \ln (1-\Delta q \cdot \theta)$ , $\theta< \frac{1}{\Delta q}$. Given $\theta = \frac{1}{2 \Delta q}$, we need to show that $\exists k, k \cdot \ln (2)> (k+1) \cdot \ln (1.5)$, 
    which always holds for all $k>1.4094$.
\end{proof}
Therefore, Gamma distribution may improve over the baseline, and this can be
computed by optimizing the privacy-utility trade-off using the Lagrange multiplier
function in Equation~\ref{lagrange1}. Also, our optimization shows that, this distribution is more effective for large $\epsilon$ (weaker privacy guarantees). 

\vspace{0.05in}
(2) \textbf{Uniform distribution.}
In probability theory and statistics, the continuous uniform distribution or rectangular distribution is a family of symmetric probability distributions such that for each member of the family, all intervals of the same length on the support of the distribution are equally probable. The support is defined by the two parameters, $a$ and $b$, which are the minimum and maximum values. The distribution is often abbreviated as $U(a,b)$, which is the maximum entropy probability distribution for a random variable $X$ under no constraint; other than that, it is contained in the distribution's support. The MGF for $U(a,b)$ is \[ M_X(t)=\begin{cases}\frac{e^{tb}-e^{ta}}{t(b-a)}&{\text{for }}t\neq 0,\\1&{\text{for }}{\text{for }}t=0.\end{cases}\] Using Theorem~\ref{simple DP}, we now drive the precise \DP guarantee of an Randomized DP Laplace mechanism for uniform distribution $U(a,b)$. 

\begin{thm}
\label{uniformdist}
 The Randomized DP Laplace mechanism $M_q(d,\epsilon)$, $\epsilon \sim f_{U(a,b)}(\epsilon)$, is $\ln \big[\frac{\alpha^2-\beta^2}{2((1+\beta)e^{-\beta}-(1+\alpha)e^{-\alpha})} \big]$-differentially private, where $\alpha=a\cdot \Delta q$ and $\beta=b\cdot \Delta q$. 
\end{thm}
We now apply the necessary condition given in Equation~\ref{necc}. One can easily verify that the inequality holds for infinite number of settings, e.g., $a=0.5$, $b=9$ and $\Delta q=1.2$.
\begin{lem}
\label{necsuniform}
   Randomized DP using uniform distribution can satisfy the necessary condition in Equation~\ref{necc}.
\end{lem}
Therefore, Randomized DP using uniform distribution may improve over the baseline, and this can be computed by optimizing the privacy-utility trade-off using the Lagrange multiplier function in Equation~\ref{lagrange1}. Also, our simulation shows that, this distribution can also be effective for both small and large $\epsilon$. 
    

\vspace{0.05in}

(3) \textbf{Truncated Gaussian distribution.} The last distribution we
consider is the Truncated Gaussian distribution. This
distribution is derived from that of a normally distributed random
variable by bounding the random variable from either below or above
(or both). Therefore, we can benefit from the numerous useful
properties of Gaussian distribution, by truncating the
negative region of the Gaussian distribution. Suppose $X\sim \mathcal
N(\mu ,\sigma ^{2})$ has a Gaussian distribution and lies within the
interval $X\in (a,b),\;-\infty \leq a<b\leq \infty$. Then, $X$
conditional on $a<X<b$ has a truncated Gaussian distribution with the
following probability density function.

\vspace{-0.15in}

\begin{equation*}
    f_{\mathcal N^T}(X;\mu ,\sigma,a,b)
    {\displaystyle={\frac {\phi(\frac{X-\mu}{\sigma})}{\sigma\cdot \big(\Phi(\frac{b-\mu}{\sigma})-\Phi(\frac{a-\mu}{\sigma})\big)}}\quad {\text{for }}a\leq x\leq b}
\end{equation*}

and by $f_{\mathcal N^T}=0$ otherwise. Here, $\phi(x)=\frac{1}{\sqrt{2\pi}\cdot} e^{-\frac{x^2}{2}}$ and $\Phi(x)=1-Q(x)$ are PDF and CDF of the standard Gaussian distribution, respectively.
Next, using Theorem~\ref{simple DP}, we give the \DP guarantee provided by the mechanism assuming that the reciprocal of $b$ is distributed according to the truncated Gaussian distribution.
\begin{thm}
\label{truncateddist}
 The Randomized DP Laplace mechanism $\mathcal M_q(d,\epsilon)$, and $\epsilon \sim f_{\mathcal N^T}(\epsilon;\mu ,\sigma,a,b)$, satisfies $\epsilon_{N^T}$-DP, where 
 \begin{align}
    \epsilon_{N^T}=\ln\left[\cfrac{\mu+\cfrac{\sigma\cdot(\phi(\alpha)-\phi(\beta))}{(\Phi(\beta)-\Phi(\alpha))}}{\diff{M_{N^T}(t)}{t}|_t=-\Delta q}\right]
 \end{align}
 in which $\phi(\cdot)$ is the probability density function of the standard normal distribution, $\phi(\cdot)$ is its cumulative distribution function and $\alpha=\frac{a-\mu}{\sigma}$ and $\beta=\frac{b-\mu}{\sigma}$.
\end{thm}
\begin{lem}
\label{necstrunc}
  Randomized DP using truncated Gaussian distribution can satisfy the necessary condition in Equation~\ref{necc}.
\end{lem}
Therefore, truncated Gaussian distribution may improve over the baseline, and this can be computed by optimizing the privacy-utility trade-off using the Lagrange multiplier
function in Equation~\ref{lagrange1}. In particular, our simulation shows that, this distribution can also be effective for smaller $\epsilon$ (stronger privacy guarantees).

\subsection{Expanding the Search Space with Combined PDFs} 
\label{compoesec}
Theorem~\ref{thm:lin} can be directly applied to design a utility-maximizing
Randomized DP Laplace mechanism with a sufficiently large search space
(infinite number of different random variables). Since the Laplace mechanism has already been studied
under $\ell_1, \ell_2$ and entropy, we will focus on the usefulness metric. 
\begin{coro}
[Optimal Utility for Combined RVs]
\label{thm:mgffin}
 If $x_1, x_2, \cdots, x_n$ are $n$ independent random variables with respective MGFs $M_{x_i}(t)=\mathbb E (e^{t x_i})$ for $i = 1, 2,\cdots, n$, then for the linear combination $Y=\sum\limits_{i=1}^{n}a_ix_i$, 
the optimal usefulness (similar relation holds for other metrics) under $\epsilon$-\DP constraint is given as 
 \begin{eqnarray}
 \scriptsize
    \label{mgffin}
       &\hspace{-1.5cm}  U_{Y}(\epsilon, \Delta q, \gamma)=1-\min\limits_{\mathcal{A,U}} \bigg\{ \prod \limits_{i=1}^{n} M_{x_i} (-a_i\gamma) \bigg\}\\
&\hspace{-7cm} \text{subject to           }  \nonumber\\ 
& \hspace{-.3cm} \epsilon= ln\left[\cfrac{\sum \limits_{j=1}^{n} a_j\cdot E_{x_j}(\frac{1}{b})}{\sum \limits_{j=1}^{n} a_j\cdot M'_{x_j}(a_j\cdot -\Delta q) \cdot \prod \limits_{\substack{i=1 \\ i\neq j}}^{n} M_{x_i} (-a_i\cdot \Delta q)}\right] \nonumber  \nonumber
    \end{eqnarray}
    where $\mathcal{A}=\{a_1,a_2,\cdots, a_n\}$ is the set of the coefficients and $\mathcal{U}=\{u_1,u_2,\cdots, u_n\}$ is the set of parameters of the probability distributions of RVs $x_i, \ \forall i \leq n$. 
\end{coro}
Similar to the case of single RVs, we can compute the optimal solution for this
optimization problem using the Lagrange multiplier function in
Equation~\ref{lagrange1}.

We will focus on all RVs that are produced using linear combinations
of the Gamma, uniform and truncated Gaussian
distributions (which include both weak and strong privacy-preserving PDFs). Therefore, the
corresponding Lagrange multiplier function is
\begin{eqnarray}
\label{objboth}
 &\hspace{-2cm}\mathcal{L} (a_1,a_2,a_3,k, \theta,a_u,b_u,\mu,\sigma,a_{\mathcal N^T}, b_{\mathcal N^T}, \Lambda) \\    &\hspace{-2cm}= M_{\Gamma(k,\theta)}(-a_1\gamma) \cdot M_{U(a_u,b_u)}(-a_2\gamma) \nonumber\\
 &\cdot M_{\mathcal{N}^T(\mu,\sigma,a_{\mathcal N^T}, b_{\mathcal N^T})}(-a_3\gamma) +\Lambda \cdot (\ln \Bigg[\cfrac{\mathsf{N}}{\mathsf{D}}\Bigg]-\epsilon)   \nonumber
\end{eqnarray}
where the numerator and the denominator $\mathsf{N, \ D}$ are
\begin{eqnarray*}
&\hspace{-7.5cm}\mathsf{N}=\\&(a_1 \cdot k \cdot \theta)+ (a_2 \cdot \frac{a+b}{2})+(a_3 \cdot (\mu+(\cfrac{\sigma\cdot\phi(\alpha)-\phi(\beta))}{(\Phi(\beta)-\Phi(\alpha))}))
\end{eqnarray*}
\begin{eqnarray*}
&\hspace{-0.7cm}\mathsf{D}=a_1 \cdot M'_{\Gamma(k,\theta)}(-a_1\cdot \Delta q) \cdot M_{U(a_u,b_u)}(-a_2 \cdot\Delta q) \\ &\cdot M_{\mathcal{N}^T(\mu,\sigma,a_{\mathcal N^T}, b_{\mathcal N^T})}(-a_3\cdot \Delta q)\\& +a_2 \cdot M_{\Gamma(k,\theta)}(-a_1\cdot \Delta q) \cdot M'_{U(a_u,b_u)}(-a_2 \cdot\Delta q) \\ &\cdot M_{\mathcal{N}^T(\mu,\sigma,a_{\mathcal N^T}, b_{\mathcal N^T})}(-a_3\cdot \Delta q)\\& +a_3\cdot M_{\Gamma(k,\theta)}(-a_1\cdot \Delta q) \cdot M_{U(a_u,b_u)}(-a_2 \cdot\Delta q) \\ &\cdot M'_{\mathcal{N}^T(\mu,\sigma,a_{\mathcal N^T}, b_{\mathcal N^T})}(-a_3\cdot \Delta q)
\end{eqnarray*}

Finally, Algorithm~\ref{alg:analyst-actions1} details our Randomized DP Laplace mechanism using linear combinations of these three   
PDFs. In Section~\ref{exp:sec}, using experiments and simulation results, we show that Algorithm~\ref{alg:analyst-actions1} can indeed outputs near-optimal results.

\label{sec:algm}

 \begin{algorithm}[!h]
 \small
 \SetKwInOut{Input}{Input}
 \SetKwInOut{Output}{Output}
\Input{Dataset $D$, Privacy budget $\epsilon$, Query $q(\cdot)$, Metric and its parameters (from data recipient)}

\Output{Query result $q(D)+Lap(b_r)$ which is $\epsilon$-DP and is near-optimal w.r.t. the utility requirement}


$\Delta q\leftarrow$ Sensitivity ($q(\cdot)$) 

Find optimal parameters from Lagrange Multiplier $\mathcal{L}(\epsilon, \Delta q,\text{metric})=$ $\{a_{1}^{opt},a_{2}^{opt}, a_{3}^{opt}, k^{opt}, \theta^{opt},a^{opt}_u,b^{opt}_u,\mu^{opt},\sigma^{opt},a_{\mathcal N^T}^{opt}, b_{\mathcal N^T}^{opt}\}$ 
 
 
$X_1 \sim \Gamma(k^{opt}, \theta^{opt})$\\
 $X_2 \sim U(a^{opt}_u,b^{opt}_u)$	\\
 $X_3\sim \mathcal{N}^T(\mu^{opt},\sigma^{opt},a_{\mathcal N^T}^{opt}, b_{\mathcal N^T}^{opt})$

 $\frac{1}{b_r} = a_1^{opt} \cdot X_1 + a_2^{opt} \cdot X_2+a_3^{opt} \cdot X_3$ 

\textbf{return} $q(D)+Lap(b_r)$
 	\caption{Randomized DP with 3 PDFs}
 	\label{alg:analyst-actions1}
 \end{algorithm}

\end{document}